\documentclass[conference,compsoc]{IEEEtran}

\ifCLASSOPTIONcompsoc
  \usepackage[nocompress]{cite}
\else
  \usepackage{cite}
\fi

\usepackage{latexsym}
\usepackage{amsmath}
\usepackage{amsxtra}
\usepackage{amssymb}
\usepackage{amsthm}
\usepackage{amsfonts}
\usepackage{amsbsy}
\usepackage{bm}
\usepackage{url} 
\usepackage{color}
\usepackage{graphicx}
\usepackage{algorithm}
\usepackage{algorithmic} 
\usepackage{rotating}
\usepackage{psfrag}
\usepackage{epstopdf}

\usepackage[table]{xcolor}

\newcommand{\taxa}{{\Sigma}} 
\newcommand{\genes}{{\Gamma}}
 
\newcommand{\genetrees}{{\mathcal G}}

\newcommand{\define}[1]{\emph{#1}}

\renewcommand{\L}{{\mathcal L}}

\renewcommand{\l}{{\ell}}

\newtheorem{corollary}{Corollary}
\newtheorem{lemma}{Lemma}
\newtheorem{property}{Property}
\newtheorem{theorem}{Theorem}
\newtheorem{definition}{Definition}

\begin{document}

%

\title{Gene Tree Construction and Correction using SuperTree and Reconciliation}
\author{\IEEEauthorblockN{Manuel Lafond}
\IEEEauthorblockA{D\'epartement d'Informatique (DIRO)\\
Universi\'e de Montr\'eal\\
Montreal, QC, Canada \\
Email: lafonman@iro.umontreal.ca}
\and
\IEEEauthorblockN{Cedric Chauve}
\IEEEauthorblockA{Department of Mathematics \\ 
Simon Fraser University\\
Burnaby, BC, Canada \\
Email: cedric.chauve@sfu.ca}
\and
\IEEEauthorblockN{Nadia El-Mabrouk}
\IEEEauthorblockA{D\'epartement d'Informatique (DIRO) \\
Universi\'e de Montr\'eal\\
Montreal, QC, Canada \\
Email: mabrouk@iro.umontreal.ca}
\and
\IEEEauthorblockN{A\"ida Ouangraoua}
\IEEEauthorblockA{D\'epartement d'Informatique\\
Universit\'e de Sherbrooke\\
Sherbrooke, QC, Canada\\
Email: aida.ouangraoua@usherbrooke.ca}
}
\maketitle
 
\begin{abstract}
The supertree problem asking for a tree displaying a set of consistent
input trees has been largely considered for the reconstruction of
species trees. Here, we rather explore this framework for the sake of
reconstructing a gene tree from a set of input gene trees on partial
data. In this perspective, the phylogenetic tree for the species
containing the genes of interest can be used to choose among the many
possible compatible ``supergenetrees'', the most natural criteria
being to minimize a reconciliation cost. We develop a variety of
algorithmic solutions for the construction and correction of gene
trees using the supertree framework.  A dynamic programming supertree algorithm for constructing or 
correcting gene trees, exponential in the number of input trees,
is first developed for the less constrained version of the problem. 
It is then adapted to gene trees with nodes labeled as duplication or
speciation, the additional constraint being to preserve the orthology
and paralogy relations between genes. Then, a quadratic time algorithm is developed for efficiently correcting an initial gene tree while preserving a set of ``trusted'' subtrees, as well as the relative phylogenetic distance between them, in both cases of labeled or
unlabeled input trees. By applying these algorithms to the set of
Ensembl gene trees, we show that this new correction framework is
particularly useful to correct weakly-supported duplication
nodes. The C++ source code for the algorithms and simulations described
in the paper are available at \verb+https://github.com/UdeM-LBIT/SuGeT+.
\end{abstract}
\IEEEpeerreviewmaketitle

\section{Introduction}
The supertree problem consists in combining a set of input
phylogenetic trees on possibly overlapping sets of data, into a single
one for the whole set (see for
example~\cite{Bininda04,Bansal10,Warnow12,PhySIC07,
  Douzery10,Steel08,Warnow12b}). Ideally, the obtained tree should
display each of the input trees,
which is only possible if they are  ``consistent'' i.e. if
they do not contain conflicting phylogenetic information.
The simplest formulation of the supertree problem is therefore to
state whether an input set of trees is consistent, and if so, find a
``compatible'' tree, called a {\em supertree}\/, displaying them
all. This problem is NP-complete for unrooted
trees~\cite{Steel92,Scornavacca14}, but solvable in polynomial time for
rooted trees~\cite{Aho81,Sankoff95,Ng96,Semple03}. However, even for
rooted trees the set of all possible supertrees may be exponential in
the number of genes.

Supertree methods have been mainly designed to reconstruct a species
tree from gene trees obtained for various gene families. However, they
can have applications for gene tree reconstruction as
well. Indeed, they may be used to combine partial trees on overlapping
gene sets available from various sources (various databases, various
reconstruction tools, \textit{etc}). Alternatively, in the case of
large gene families, they may be used to combine gene trees for
smaller sets of orthologs, usually obtained from clustering algorithms
such as OrthoMCL~\cite{OrthoMCL03}, InParanoid~\cite{InParanoid08} or
Proteinortho~\cite{Proteinortho11}.  In such a case, ideally,
orthology relations should be preserved in the final tree. More
generally, given a set of input ``labeled gene trees'', i.e. gene
trees with internal nodes labeled as duplication or speciation, we may
be interested in a supertree preserving this labeling. As far as we know, no automated method accounting for labeling constraints has never been proposed. Here, we consider the problem of reconstructing a ``supergenetree'' in both cases of a labeled or unlabeled set of input gene trees.

In this paper, we also show that the supertree principle can be used
for gene tree correction. For various reasons related to the
considered model, method or data, gene trees can contain many errors (see for example~\cite{HAHN07} for a link with dubious high duplication nodes), and trees frequently exhibit branches with low statistical
support. Two main approaches exist to correct gene trees, based on
a local exploration principle to identify closely related trees that might have a better
statistical support~\cite{TREEFIX}, a better reconciliation
cost~\cite{CHEN-JCOMPBIOL7,Zheng-Zhang2014,SwensonMabrouk12} or a combination of
both~\cite{TREEFIXDTL,noutahi2016}. In the present work, we consider the second
approach, based on the reconciliation cost with a given species tree.
A way of correcting a gene tree is to remove weakly-supported
branches, leading to a set of subtrees, that should then be merged  into a new one, according to some criterion. The most
commonly considered criterion is a best fit with the species tree. A
simple way is to consider the set of subtrees as the leaves of a
polytomy (star-tree), and to resolve the polytomy in a way minimizing
the reconciliation cost with the species
tree (see NOTUNG~\cite{CHEN-JCOMPBIOL7}, the Zheng and Zhang algorithm~\cite{Zheng-Zhang2014}, PolytomySolver~\cite{LafondNoutahi2016}). Such a
correction method, not only preserves the input subtrees, but also the
gene clusters inside the subtrees. In other words, the exhibited
monophily of input gene clusters is not challenged by a polytomy
resolution method. However, it has been shown that genes under
negative selection, while exhibiting the true topology, may be wrongly
grouped into monophyletic groups (see for example
~\cite{Liberles08,Skovgaard06,Taylor05,SwensonMabrouk12}). In this
perspective, using a supertree method may be beneficial, as it
preserves the topology of subtrees while allowing to group
genes from different subtrees.

In~\cite{Lafond15}, we introduced under the name of {\sc Minimum
  SuperGeneTree} ($MinSGT$) the problem of finding, for a set of gene trees,
a supertree that minimizes the reconciliation cost with a given
species tree. Under the duplication cost, we have shown that this
problem is NP-hard to approximate within a $n^{1 - \epsilon}$ factor,
for any $0 < \epsilon < 1$, even for instances in which there is only
one gene per species in the input trees, and even if each gene appears
in at most one input tree.  In this paper, we carry out on $MinSGT$ but for the more general
reconciliation cost. Although
NP-hardness proofs for the duplication cost do not apply to the
duplication plus loss cost, the problem is conjectured NP-hard for
this more general reconciliation cost, as accounting for losses in
addition to duplications is unlikely to make the problem
simpler. Given a set of consistent input gene trees, we provide
various algorithmic results depending on the additional information we
have on the trees.  

In Section~\ref{section-minSGT}, we first exhibit
a dynamic programming algorithm for the general case, exponential in
the number of input trees. We show how this algorithm can be adapted
to 
compute a supertree preserving the input trees labeling, as motivated above.
In Section~\ref{triplets}, we then consider the correction problem with as input a gene tree together with a set
of subtrees which topology should be preserved in the final
supertree. To avoid having a supertree grouping genes that are far
apart in the original tree, the relative phylogenetic distance between
gene clusters is considered as an additional constraint.  Inpired by the {\sc Minimum Triplet Respecting History} introduced in \cite{SwensonMabrouk12}, we
define the {\sc Minimum Triplet Respecting SuperGeneTree Problem}
asking for a supertree displaying all input subtrees, while preserving
the topology of any triplet of genes taken from three different
subtrees. We
develop a quadratic-time algorithm for this problem.  Finally, in
Section~\ref{application}, by applying these algorithms to a set of a few hundreds
Ensembl vertebrate gene trees, we show that this new correction framework is
particularly useful to correct weakly-supported upper duplication
nodes, as we observe that the correction carried out by our
algorithms often improves significantly the likelihood scores.

\section {Preliminaries}

All considered trees are rooted and binary.  We denote by $r(T)$ the
root, by $V(T)$ the set of nodes, and by $\L(T) \subseteq V(T)$ the
leafset of a tree $T$. We say that $T$ is a tree for $L=\L(T)$. Given
a node $x$ of $T$, the subtree of $T$ rooted at $x$ is denoted
$T[x]$. When there is no ambiguity on the considered tree, we simply
write $\L(x)$ instead of $\L(T[x])$.  We arbitrarily set one of the
two children of an internal node $x$ as the left child $x_l$ and the
other as the right child $x_r$, and denote by $(\L(x_l), \L(x_r))$ the
bipartition induced by $x$. Also for the sake of simplicity, we just
denote by $T_l$ and $T_r$ the left and right subtrees of the root of
$T$. A node $x$ is an \emph{ancestor} of a node $y$ if $x$ is on the
path between $y$ and $r(T)$.
If $x$ is an ancestor of $y$,
$inter(x,y)$ is the number of nodes located on the path between $x$
and $y$, excluding $x$ and $y$.  Two nodes $x$ and $y$ are {\em
  separated}\/ in $T$ iff none is an ancestor of the other. In this
case, we also say that the two subtrees $T[x]$, $T[y]$ of $T$ are {\em
  separated}\/.

The \emph{lowest common ancestor} (lca) of $L' \subset \L(T)$,
denoted $lca_T(L')$, is the ancestor common to all leaves in $L'$ that
is the most distant from the root.  $T|_{L'}$ is the tree with leafset
$L'$ obtained from the subtree of $T$ rooted at $lca_T(L')$ by
removing all leaves that are not in $L'$, and then all internal nodes
of degree 2, except the root.  Let $T'$ be a tree such that $\L(T') =
L' \subseteq \L(T)$. We say that $T$ {\em displays}\/ $T'$ iff
$T|_{L'}$ is isomorphic to $T'$ while preserving the same leaf-labeling.\\

\begin{figure*}[!t]
\begin{center}
\includegraphics[width= .85 \textwidth]{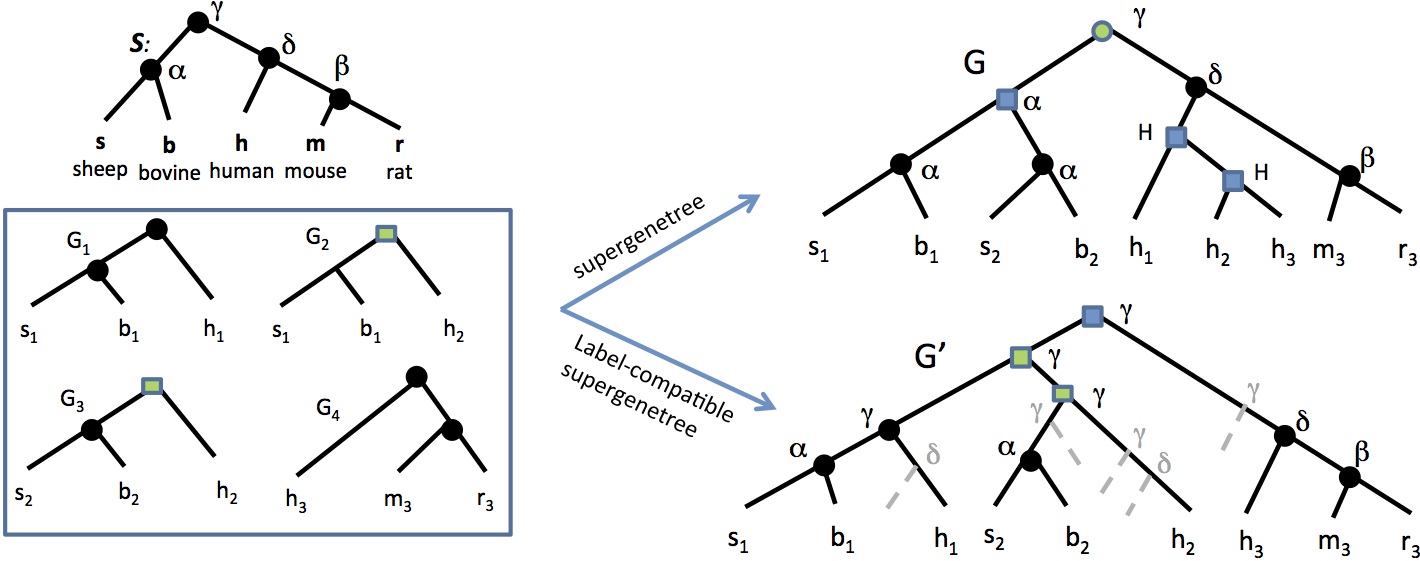}\hspace{1cm}
\caption{A species tree $S$ on $\Sigma = \{s,b,h,m,r\}$, and a set of
  labeled gene trees $\genetrees = \{G_1,G_2,G_3,G_4\}$, where each
  leaf $x_i$ denotes a gene belonging to $x$. Square nodes are
  duplications and circular nodes are speciations. Internal nodes are
  labeled according to corresponding ancestral species in
  $S$. Dotted lines are losses.  $G$ is a supergenetree for
  $\genetrees$ of minimum LCA-reconciliation cost (cost of $3$), while
  $G'$ is a label-compatible supergenetree for $\genetrees$ of
  reconciliation cost $8$ ($3$ duplications + $5$ losses). $G$ is not
  a label-compatible supergenetree due to the roots of $G_2$ and $G_3$
  which are duplications in $G_2$ and $G_3$ (green squares), but are
  mapped to a speciation node in $G$ (green circle). In $G'$, these
  nodes are correctly mapped to duplication nodes (green squares in
  $G'$).}\label{fig:supergenetree}
\end{center}
\end{figure*}

\noindent\textit{Gene and species trees.} 
A species tree $S$ for a set $\taxa$
of species represents an ordered set of speciation events that have
led to $\Sigma$.
A \emph{gene family} is a set of genes $\genes$ accompanied with a
\emph{mapping function} $s : \genes \rightarrow \taxa$ mapping each
gene to its corresponding species. Consider a gene family $\genes$
where each gene $x\in \genes$ belongs to a species $s(x)$ of $\taxa$.
The evolutionary history of $\genes$ can be represented as a {\em gene
  tree}\/ $G$ for $\genes$. For example, in
Figure~1, $G$ is a gene tree for $\genes =
\{s_1, s_2, b_1, b_2, h_1, h_2, h_3, m_3, r_3\}$. Each internal node
of $G$ refers to an ancestral gene at the moment of an event, either
speciation ($Spec$) or duplication ($Dup$). The mapping function $s$ is
generalized as follows: if $x$ is an internal node of $G$, then
$s(x) = lca_S(\{s(x') : x' \in \L(x) \})$. 

When the type of event is known for each internal node, the gene tree
$G$ is said {\em labeled}\/. Formally, a {\em labeled gene tree}\/ for
$\genes$ is a pair $(G, ev_G)$, where $G$ is a tree for $\L(G) =
\genes$, and $ev_G:V(G) \setminus \L(G) \rightarrow \{Dup, Spec\}$ is
a function labeling each internal node of $G$ as a duplication or a
speciation node.

According to the Fitch~\cite{Fitch2000} terminology, given a labeled
gene tree $(G, ev_G)$, we say that two genes $x, y$ are
\emph{orthologs} if $ev_G(lca_G(x, y)) = Spec$, and
\emph{paralogs} if $ev_G(lca_G(x, y)) = Dup$. For example, from the
set of labeled gene trees in Figure~1, $s_1,
h_1$ are orthologs while $s_1, h_2$ are paralogs.

While a history for $\genes$ can be represented as a labeled gene
tree, the converse is not always true, as a labeled tree $(G, ev_G)$
for $\genes$ does not necessarily represent a valid history in
agreement with a species tree $S$. For this to hold, $(G, ev_G)$
should be \emph{$S$-consistent}, i.e. any speciation node of $(G,
ev_G)$ should reflect the same clustering of species as in $S$
(see~\cite{lafond2014orthology} for a formal definition of
$S$-consistency). \\


\noindent\textit{Reconciliation.} 
The {\em LCA-reconciliation}\/ of
$G$ with $S$ is the labeled tree $(G, ev_G)$ obtained by labeling each
node $x$ of $G$ as $Spec$ if and only if $s(x_l)$ and $s(x_r)$ are
separated in $S$, and as $Dup$ otherwise. It follows that the
LCA-reconciliation of $G$ with $S$ is an $S$-consistent tree. In
Figure~1, $G$ is labeled according to
the LCA-reconciliation.

Given a labeled gene tree $(G, ev_G)$, the {\em duplication cost}\/ of
$(G, ev_G)$ is its number of duplication nodes.  It reflects the
number of duplications required to explain the evolution of the gene
family inside the species tree $S$ according to $G$.  A well-known
reconciliation approach~\cite{Zhang97,CHEN-JCOMPBIOL7,CHAUVE-ELMABROUK09}
allows to further recover, in linear time, the minimum number of
losses underlined by such an evolutionary history. We refer to the
number of duplications and losses underlined by
a labeled gene tree
as its {\em reconciliation cost}\/ in the general
case, and as its {\em LCA-reconciliation cost} if the tree is labeled
according to the LCA-reconciliation. \\


\noindent\textit{Supertree problems.}
Given a set $\genetrees$ of trees for possibly overlapping subsets of
$\genes$, the goal is to
find a single tree displaying them all. This is possible only if the
input trees are pairwise {\em consistent}\/. The consistency problem
of rooted trees has been largely studied. For trees to be consistent,
each triplet of data should exhibit the same topology in all
trees. The BUILD algorithm~\cite{Aho81} can be used to test, in
polynomial-time, whether a collection of rooted trees is consistent,
and if so, construct a compatible, not necessarily fully resolved,
supertree. This algorithm has been generalized to output all
compatible supertrees~\cite{Sankoff95,Ng96,Semple03}, which may be
exponential in the number of genes.

\section{Algorithms for Minimum SuperGeneTree Problems}\label{section-minSGT}

We begin with the less constrained version of the problem. Given a set
$\genetrees$ of consistent input gene trees, we ask for a {\em
  compatible}\/ tree, also called \emph{supergenetree} $G$ for
$\genetrees$, i.e. a tree displaying each tree of $\genetrees$. In
addition, among all supergenetrees for $\genetrees$, $G$ should be of
minimum LCA-reconciliation cost (see $G$ in
Figure~1).\\

\noindent \textsc{Minimum SuperGeneTree ($MinSGT$) Problem:}\\
\noindent {\bf Input:} A species set $\taxa$ and a species tree $S$
for $\taxa$; a gene family $\genes$ of size $n$, a set $\genes_{i, 1 \leq i \leq
  k}$ of subsets of $\genes$ such that $\bigcup_{i=1}^{k}{\genes_i} = \genes$,
and a consistent set $\genetrees = \{G_1, G_2,\cdots, G_k \}$
of gene trees such that, for each $1 \leq i
\leq k$, $G_i$ is a tree for $\genes_i$. \\
\noindent {\bf Output:} Among all trees $G$ for $\genes$ compatible
with $\genetrees$, one of minimum LCA-reconciliation cost.\\

Suppose now that the input trees are labeled, and consider this
labeling as an additional constraint. The problem becomes one of
finding a labeled supergenetree preserving the input gene trees node
labeling. As a labeled gene tree induces a full orthology and paralogy
relation on the set of its leaves, this is possible only if the set of
relations is {\em satisfiable}\/, i.e. if there is a labeled tree $(G,
ev_G)$ displaying the relations induced by all the input trees, and if
there is such a tree which is $S$-consistent. Satisfiability is a
well-studied problem. It reduces to verifying if a relation graph $R$
(vertices are genes and edges link orthologous genes) is $P_4$-free,
i.e. no four vertices of $R$ induce a path of length
$3$~\cite{Hellmuth-2013}. On the other hand, a cubic-time algorithm
was developed in~\cite{lafond2014orthology} for deciding whether
a set of relations is $S$-consistent. Hereafter, we assume that the
relations induced by the input trees are satisfiable and
$S$-consistent.

Let $G$ and $G'$ be two trees with $\L(G') \subseteq \L(G)$ such that
$G$ displays $G'$. Then $(G,ev_G)$ is said {\em label-compatible} with
$(G', ev_{G'})$ iff, for any internal node $x$ of $G$ and $x'$ of $G'$
such that $x = lca_{G}(\L(x'))$, $ev_G(x) = ev_{G'}(x')$. A labeled
supergenetree $G$ for a set $\genetrees$ of trees is said
\define{label-compatible} with $\genetrees$ iff it is label-compatible
with each of the labeled trees of $\genetrees$. An illustration is
provided by the supergenetree $G'$ in
Figure~1. We are now ready to formulate our
second problem. \\

\noindent \textsc{Minimum Labeled SuperGeneTree ($MinLSGT$) Problem:}\\
\noindent {\bf Input:} A species set $\taxa$ and a species tree $S$
for $\taxa$; a gene family $\genes$ of size $n$, a set $\genes_{i, 1 \leq i \leq
  k}$ of subsets of $\genes$ such that $\bigcup_{i=1}^{k}{\genes_i} = \genes$,
and a consistent set $\genetrees = \{(G_1,ev_1),
(G_2,ev_2), \cdots, (G_k,ev_k) \}$ of satisfiable and $S$-consistent 
labeled gene trees where, for each $1 \leq i \leq k$, $G_i$ is a tree
for $\genes_i$.\\
\noindent {\bf Output:} Among all labeled supergenetrees $(G, ev_G)$
for $\genes$ label-compatible with $\genetrees$, one of minimum
reconciliation cost.\\

The $MinSGT$ and $MinLSGT$ problems for the duplication cost were both shown NP-Hard  in~\cite{Lafond15}, even in the case where no two input trees have a gene in common and the trees only contain speciations.

\subsection{The $MinSGT$ problem}
We describe a dynamic programming algorithm
for the $MinSGT$ problem leading to the following result.

\begin{theorem}\label{theorem:complexity-minsgt}
The $MinSGT$ problem can be solved in $O((n+1)^k \times 4^k \times k)$ time complexity.
\end{theorem}

The algorithm constructs the supergenetree $G$ from the root to the
leaves. At each step, i.e. for each internal node $x$ being
constructed in $G$, all possible bipartitions
$(\L(x_l),\L(x_r))$ that could be induced by $x$ are tried, 
and the iteration continues on each of $\L(x_l)$ and $\L(x_r)$.  
For example, at the root, the goal is to find the
best bipartition of $\genes$, i.e. the one leading to the minimum
LCA-reconciliation cost. At each step, this cost is computed from a {\em
  local reconciliation cost}\/ at $x$ (as defined in
Lemma~\ref{lem:reconciliation-cost}), and from the best reconciliation
cost of the two created clusters. A key observation is that the
constraint of being compatible with the input gene trees induces a
strong constraint on the bipartitions, hence only a subset of
the bipartition set has to be tested at each step.



\begin{figure*}[!t]
\begin{center}
\includegraphics[width= 15cm]{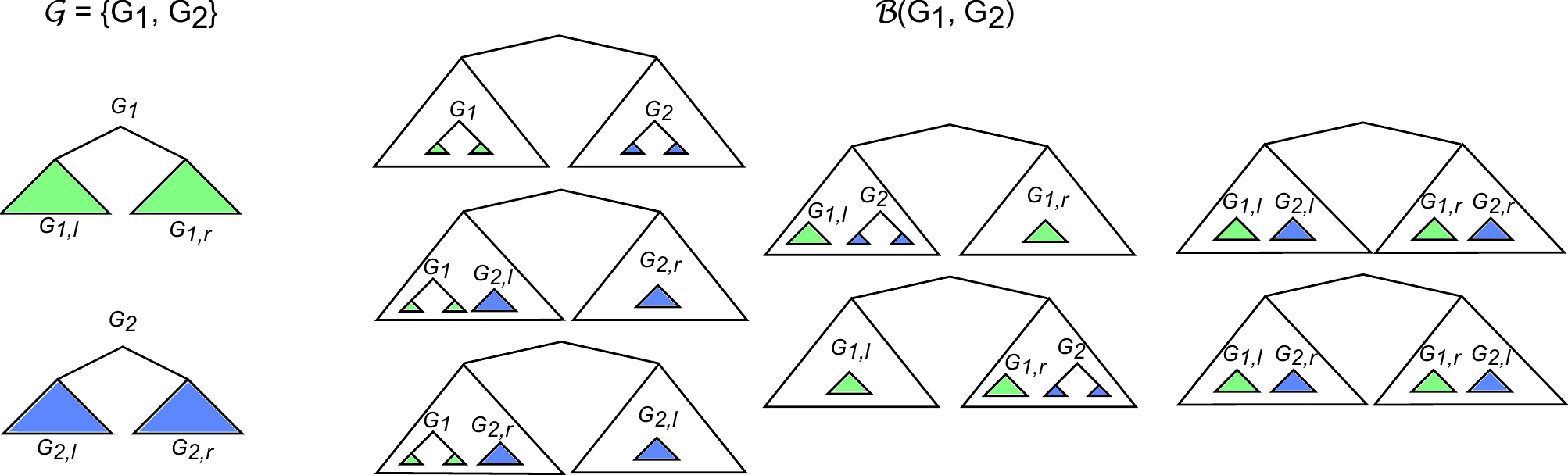}\hspace{1cm}
\caption{An illustration of the seven valid bipartitions for two trees $G_1$ and $G_2$.  Each bipartition is obtained by ``sending'' $L_1 \in \{\L(G_1), \L(G_{1, l}), \L(G_{1, r}), \emptyset\}$ in the left part, and the complement $\L(G_1) \setminus L_1$ in the right part.  The same process is then applied to $G_2$.  The set $\mathcal{B}(G_1, G_2)$ consists in the set of all possible combinations of choices, after eliminating symmetric cases and partitions with an empty side.  For each bipartition $(L_l, L_r)$ of $\mathcal{B}(G_1, G_2)$, an optimal solution is recursively computed for $L_l$ and $L_r$, the reconciliation cost is computed for 
the tree obtained by joining the roots of the two trees under a common parent, and the tree  yielding a minimum cost among all possibilities is returned.}\label{fig:bipartitions}
\end{center}
\end{figure*}

First, a formulation of the reconciliation cost in terms of the
sum of local reconciliation costs at each internal node $x$ is given.
The next lemma is a
reformulation of the reconciliation cost, as described in many
papers~\cite{CHEN-JCOMPBIOL7,CHAUVE-ELMABROUK09}.

\begin{lemma}
  \label{lem:reconciliation-cost}
  The LCA-reconciliation cost of a gene tree $G$ is the sum of {\em
    local LCA-reconciliation costs}\/ $cost(L_l,L_r)$ for all internal
  nodes $x$ of $G$, where $L=\L(x)$, and $(L_l,L_r) =
  (\L(x_l),\L(x_r))$, and $cost(L_l,L_r)$ equals to:
\begin{itemize}
\item $inter(s(L),s(L_l)) +  inter(s(L),s(L_r))$\\ if $s(L) \neq s(L_l)$ and $s(L) \neq s(L_r)$;
\item $1 + inter(s(L),s(L_l)) +  inter(s(L),s(L_r))$\\ if $s(L)= s(L_l)$ and $s(L) = s(L_r)$;
\item $2 + inter(s(L),s(L_l)) + inter(s(L),s(L_r))$\\ if $s(L)= s(L_l)$
  and $s(L)\neq s(L_r)$ or conversely.
\end{itemize}

The node $x=(L_l,L_r)$ is a speciation node in the first case, and a
duplication node in the two last cases (thus adding $1$ duplication to
the LCA-reconciliation cost, plus $1$ loss in the third case).  Note
that $inter(s,t)=0$ if $s=t$.
\end{lemma}

For example, the root of $G$ in Figure~1
fulfills the conditions of the first case, and thus it is a speciation
node, whereas the root of $G'$ fulfills the condition of the third
case.

Lemma \ref{lem:reconciliation-cost} allows to recursively compute
a minimum LCA-reconciliation cost supergenetree, by exploring, for
each node $x$ from the root to the leaves, all ``valid'' bipartitions
of $\L(x)$, remaining to be characterized formally. In the following,
we define the properties of a bipartition $(L_l,L_r)$ induced by the
root of a supergenetree $G$. It directly follows from the definition
of a supergenetree that should display each individual gene tree.

\begin{property}
  \label{prop:valid-bipartition}
  Let $\genetrees = \{G_1,\ldots,G_k\}$ be a set of gene trees.  The
  root of a supergenetree $G$ compatible with $\genetrees$ subdivides
  $\bigcup_{i=1}^{k}{\L(G_i)}$ into a {\em compatible bipartition}\/
  $(L_l,L_r)$, i.e. a bipartition such that, for each $i$ s.t. $1 \leq
  i \leq k$, either: 1) $\L(G_i) \subseteq L_l$; or 2) $\L(G_i)
  \subseteq L_r$; or 3) $\L(G_{i_l}) \subseteq L_l$ and
    $\L(G_{i_r}) \subseteq L_r$; or 4) $\L(G_{i_l}) \subseteq
        L_r$ and $\L(G_{i_r}) \subseteq L_l$.
\end{property}

For example, the root of the supergenetree $G$ in
Figure~1 satisfies the third condition for
$G_1$, $G_2$ and $G_3$, and the second for $G_4$.
%

$\mathcal{B}(G_1,\ldots,G_k)$ denotes the set of all bipartitions
of $\bigcup_{i=1}^{k}{\L(G_i)}$ compatible with $\genetrees$. For
example, the two bipartitions defined by the roots of $G$ and $G'$ in
Figure~\ref{fig:supergenetree} are both compatible with the given set
of gene trees.
Figure~\ref{fig:bipartitions} illustrates the set of 
all valid bipartitions compatible with two given trees.

\begin{lemma}\label{lemme:count-partitions}
$|\mathcal{B}(G_1,\ldots,G_k)| \leq (\frac{4^k}{2})-1$.
\end{lemma}

\begin{proof}
For each tree $G_i$, there are four possibilities for placing
$\L(G_{i_r})$ and $\L(G_{i_l})$ in a bipartition $(L_l,L_r)$: either
they are both in $L_l$, or both in $L_r$, or one in $L_l$ and the
other in $L_r$. Therefore, $4^k$ distributions of left and right
subtrees of the $k$ trees in $(L_l,L_r)$. However, as the left and
right characterization of nodes is arbitrary, each distribution is
counted twice, and thus the total number of different bipartitions is
$\frac{4^k}{2}$. One of these bipartitions has a part that is empty.
We discard it and the total number is then  $\frac{4^k}{2}-1$.
However, a set $(L_l,L_r)$ obtained
from such distribution of the $G_i$ subtrees is not necessarily a
bipartition, as a same gene can be present in two different input
trees, and end up placed in both $L_l$ and $L_r$. Therefore,
$(\frac{4^k}{2})-1$ is only an upper bound of the number of compatible
bipartitions.
\end{proof}

The constructive proof of Lemma~\ref{lemme:count-partitions} induces
an algorithm for enumerating the members of $\mathcal{B}(G_1,\ldots,G_k)$, which is illustrated in Figure~\ref{fig:bipartitions} for the case of two trees.
Intuitively, to construct a bipartition $(L_l, L_r)$, each tree $G_i$ of $\mathcal{G}$ can choose to ``send'' 
in $L_l$ either its left subtree $G_{i_l}$, its right subtree $G_{i_r}$, 
the whole tree $G_i$ or nothing at all.  
What has not been sent in $L_l$ is sent in $L_r$.
Then $\mathcal{B}(G_1, \ldots, G_k)$ is the set of all possible combinations of choices.
However, not every bipartition constructed in this manner
yields a valid bipartition.  For instance in Figure~\ref{fig:bipartitions}, the top-left bipartition 
cannot be valid if $G_1$ and $G_2$ share a leaf with the same label, as a gene cannot be sent both left and right.
These cases, however, can be detected easily by verifying 
the sizes of $L_l$ and $L_r$.

We are now ready to give the main recurrence formula of our dynamic
programming algorithm. Denote by $MinSGT(G_1,\ldots,G_k)$ the
minimum LCA-reconciliation cost of a supergenetree compatible with
$\genetrees =
\{G_1,\ldots,G_k\}$. The next lemma directly follows from Lemma~\ref{lem:reconciliation-cost} and Property~\ref{prop:valid-bipartition}.

\begin{lemma}\label{theorem:recurrences}
  Let $\genetrees = \{G_1,\ldots,G_k\}$ be a set of gene trees.

  \begin{enumerate}
\item $MinSGT(G_1,\ldots,G_k) = 0$ if $|~\bigcup_{i=1}^{k}{\L(G_i)}~| = 1$ (Stop condition);
\item Otherwise, 
  $ MinSGT(G_1,\ldots,G_k) =$
  \end{enumerate}
$$ \min_{(L_l,L_r)\in \mathcal{B}(G_1,\ldots,G_k)} 
  \left\{
  \begin{array}{l} cost(L_l,L_r) +\\  MinSGT(G_{1|L_l},\ldots,G_{k|L_l}) +\\  MinSGT(G_{1|L_r},\ldots,G_{k_|L_r})\end{array}
\right\}
$$
\end{lemma}

%
%

Note that, given a bipartition $(L_l,L_r) \in \mathcal{B}(G_1,\ldots,G_k)$,
for each $i$ such that $1 \leq i \leq k$, $G_{i|L_l}$
and $G_{i|L_r}$ are equal either to $\emptyset$ or $G_i$ or $G_{i_l}$
or $G_{i_r}$. Thus, $G_{i|L_l}$ and $G_{i|L_r}$ are either empty trees
or complete subtrees of $G_i$.

Note also that, at each step, the existence of a compatible
bipartition follows from the fact that the input gene trees are
assumed to be consistent, as stated in the formulation of the $MinSGT$ problem. In the absence of this assumption, we have to
add a third equation to Lemma~\ref{theorem:recurrences}: If
$|~\bigcup_{i=1}^{k}{\L(G_i)}~| > 1$ and
$|~\mathcal{B}(G_1,\ldots,G_k)~| = 0$, $MinSGT(G_1,\ldots,G_k) =
+\infty$.
\\

\noindent\textit{Complexity.} We now address the complexity of the dynamic 
programming algorithm defined by the recurrences of
Lemma~\ref{theorem:recurrences}. Each call to the recursive
procedure $MinSGT$ receives as input at most one subtree from each
tree $G_i$. Let $n$ be the maximum number of node in a tree $G_i$.
As each tree has at most $n$ possible subtrees, there are at most
$(n+1)^k$ possible calls to $MinSGT$. Next, for any set of
gene trees $\{G_1,\ldots,G_k\}$, the number of distinct bipartitions
$(L_l,L_r) \in \mathcal{B}(G_1,\ldots,G_k)$ to be tested is at most
$\frac{4^k}{2}-1$ (Lemma~\ref{lemme:count-partitions}). Finally, the
value of $cost(L_l,L_r)$ can be computed in time $O(k)$ provided that
the mapping $s$ is precomputed for all nodes of the trees
$G_1,\ldots,G_k$, and $lca(x,y)$ and $inter(x,y)$ are precomputed for
any pair $(x,y)$ of nodes in $S$. The time complexity of the overall
algorithm is therefore $O((n+1)^k \times 4^k \times k)$, which completes the proof of Theorem~\ref{theorem:complexity-minsgt}.

\subsection{The $MinLSGT$ Problem}
The algorithm for the $MinSGT$ problem can be adapted to solve
the $MinLSGT$ problem, leading to the following result.

\begin{corollary}\label{theorem:complexity-minlsgt-core}
The $MinLSGT$ problem can be solved in $O((n+1)^{k} \times 4^{k} \times k)$ time complexity.
\end{corollary}

The intuition behind the $MinLSGT$ algorithm is quite simple.
We proceed as in the $MinSGT$ algorithm, but each time a
bipartition $(L_l, L_r)$ is considered, we verify whether the root of a tree separating $L_l$ and $L_r$ should be
a  speciation or a duplication.  If there are two genes 
$g_l \in L_l$ and $g_r \in L_r$ that disagree with this 
event, we treat $(L_l, L_r)$ as an invalid bipartition
and do not consider it further.

Before describing this adaptation, we
need few additional definitions and properties. Given a set of
labeled gene trees $\genetrees =
\{(G_1,ev_{G_1}),\ldots,(G_k,ev_{G_k})\}$ and a bipartition
$(L_l,L_r)\in \mathcal{B}(G_1,\ldots,G_k)$, for any $i$ s.t. $1 \leq i
\leq k$, we say that $G_i$ is {\em separated}\/ by $(L_l,L_r)$ iff
$G_i$ satisfies the third or fourth condition of Property
\ref{prop:valid-bipartition}.
We denote by $\genetrees(L_l,L_r)$ the set of gene trees $G_i$,
$1 \leq i \leq k$, that are separated by $(L_l,L_r)$.

\begin{lemma}
\label{prop:valid-bipartition-labeled}
Let $\genetrees = \{(G_1,ev_{G_1}),\ldots,(G_k,ev_{G_k})\}$ be a set
of labeled gene trees. Then, for any labeled supergenetree $(G,ev_G)$
label-compatible with $\genetrees$, the label $ev_G(x)$ of its root
$x$ equals the label of the root of any gene tree $G_i$, $1 \leq i
\leq k$, such that $G_i\in \genetrees(\L(x_l),\L(x_r))$.
\end{lemma}

\begin{proof}
  Let $G_i$, $1 \leq i \leq k$ be a genetree such that
  $G_i \in \genetrees(\L(x_l), \L(x_r))$, and let $x_i$ be the root of $G_i$.
  Then $lca_G(\L(x_i)) = x$, and thus by definition of the label-compatibility
  of $G$ with $G_i$, we have $ev_G(x) = ev_{G_i}(x_i)$.
\end{proof}

From Lemma \ref{prop:valid-bipartition-labeled}, we define a
bipartition of $\bigcup_{i=1}^{k}{\L(G_i)}$ label-compatible with
$\genetrees$ as follows.

\begin{definition}
  \label{def:label-compatible}
  Let $\genetrees = \{(G_1,ev_{G_1}),\ldots,(G_k,ev_{G_k})\}$ be a set
  of labeled gene trees.
  A bipartition $(L_l,L_r)$ of $\bigcup_{i=1}^{k}{\L(G_i)}$ is {\em
    label-compatible}\/ with $\genetrees$ if it is compatible with
  $\genetrees$ and verifies:
\begin{enumerate}
\item  if $|~\genetrees(L_l,L_r)~| > 0$, the roots of all gene trees in
  $\genetrees(L_l,L_r)$ have the same label denoted by $ev_{\genetrees(L_l,L_r)}$.
\item if $|~\genetrees(L_l,L_r)~| > 0$ and $ev_{\genetrees(L_l,L_r)} = Spec$,
  then $lca_S(\{s(x) : x \in L_l \cup L_r \}) \neq lca_S(\{s(x) : x \in L_l \})$
  and $lca_S(\{s(x) : x \in L_l \cup L_r \}) \neq lca_S(\{s(x) : x \in L_r \})$.
\end{enumerate}
\end{definition}

For example, the bipartition determined by the root of the
supergenetree $G$ $(\{s_1, s_2, b_1, b_2\},\{h_1, h_2, h_3, m_3,
r_3\})$ in Figure~1 is not label-compatible with
the set of gene trees, as it separates both $G_1$ and $G_2$ which do
not have the same root label.

The $MinLSGT$ algorithm for solving the $MinSGT$ problem is
based on the same general dynamic programming framework as the
$MinSGT$ algorithm: at each step, iterate over all possible
bipartitions, and then proceed recursively for each partition. The two
differences are: (1) given a set of labeled gene trees
$\genetrees=\{(G_1, ev_1),\ldots,(G_k, ev_k)\}$, we only test a subset
of compatible bipartitions of $\mathcal{B}(G_1,\ldots,G_k)$ that are
label-compatible with $\genetrees$; (2) computing local reconciliation costs should not
be done on the basis of the LCA-reconciliation, as some nodes that
would be labeled as speciation nodes from the LCA-mapping should
rather be duplication nodes in order to be label-compatible with some
input gene trees. For example, in Figure~1,
$lca_{G'}(\{s_2,b_2,h_2\})$ would be labeled $Spec$ by the
LCA-mapping. However, it should be labeled $Dup$ to be
label-compatible with $G_3$. The following Lemma is
required, in place of Lemma~\ref{lem:reconciliation-cost}.

%

\begin{lemma}
 \label{lem:reconciliation-cost2}  
Let $x$ be an internal node of a labeled supergenetree $(G,ev_G)$, $L= \L(x)$ and $(L_l,L_r)=(\L(x_l),\L(x_r))$.  The local reconciliation cost of $x$, $cost_{\genetrees}(L_l,L_r)$ is equal to:

\begin{itemize}
\item $3 + inter(s(L),s(L_l)) +  inter(s(L),s(L_r))$ if $s(L) \neq s(L_l)$, $s(L) \neq s(L_r)$, $|~\genetrees(L_l,L_r)~| > 0$  and $ev_{\genetrees}(L_l,L_r) = Dup$;
\item $cost(L_l,L_r)$ Otherwise.
\end{itemize}

In the first case, the node $x$ is a duplication node adding $1$ duplication
plus at least $2$ losses to the reconciliation cost, and in the second case
the local reconciliation cost is computed as for the LCA-reconciliation.
\end{lemma}

The complexity of the $MinLSGT$ algorithm remains in $O((n+1)^k \times 4^k
\times k)$
provided that 
the sets of label-compatible bipartitions $(L_l,L_r)$ are constructed
simultaneously with the sets $\genetrees(L_l,L_r)$.

\subsection{Improved complexity from a core set of trees}
Last, we show that the principles underlying the two algorithms described above can be improved to reduce the dependency in $k$.  
The key remark is that all bipartitions to consider can be identified by considering only a subset of the input trees provided they span the set of all genes of $\genes$. 

Call $\genetrees' \subseteq \genetrees$ a \emph{core}  of $\genetrees$ if $\bigcup_{G \in \genetrees'} \L(G) = \genes$.
We introduce the following modified $MinSGT$ algorithm, that we call $MinSGT\textnormal{-}core$:
\begin{enumerate}
\item Find a core $\genetrees' = \{G'_1, \ldots, G'_{\l}\}$ of $\genetrees$.
\item Apply the $MinSGT$ algorithm on $\genetrees'$, with the exception that, when considering a bipartition $B = (L_l, L_r)$ compatible with $\genetrees'$:\\
$\bullet$ Verify that $B$ is also compatible with $\genetrees \setminus \genetrees'$.  If not, then do not proceed recursively on $B$;\\
$\bullet$ Compute $cost(L_l,L_r)$ on the whole set $\genetrees$.
\end{enumerate}

\begin{theorem}\label{theorem:complexity-minsgt-core}
Let $\genetrees'$ be a core of $\genetrees$ composed of $k'$ trees. 
The $MinSGT$ problem can be solved in $O((n+1)^{k'} \times 4^{k'} \times k)$ time complexity.
\end{theorem}

\begin{proof}
The difference between the executions of a call of the $MinSGT\textnormal{-}core$ algorithm on the input $\genetrees'$ and a call of the $MinSGT$ on $\genetrees$ lies in the set of bipartitions
considered at each step of the recursion. At a given step of the recursion, let ${\cal B}'$ and ${\cal B}$ be the set of bipartitions compatible with  $\genetrees'$ and $\genetrees$ respectively, and let ${\cal B}^*$ be the set of bipartitions considered by $MinSGT\textnormal{-}core$. We show that ${\cal B} = {\cal B}^*$.



 Clearly ${\cal B} \subseteq {\cal B}'$, as a bipartition compatible with $\genetrees$ is also compatible with $\genetrees' \subseteq \genetrees$.  
 Suppose that there is some $(L_l, L_r) \in {\cal B}$ such that 
 $(L_l, L_r) \notin {\cal B}^*$.  This implies that $(L_l, L_r)$ 
 was filtered out of ${\cal B}'$, meaning that it is not compatible with some tree $G \in \genetrees \setminus \genetrees'$.  Therefore $(L_l, L_r)$ cannot be in ${\cal B}$, a contradiction.  We deduce that ${\cal B} \subseteq {\cal B}^*$.
To see that ${\cal B}^* \subseteq {\cal B}$, observe that ${\cal B}^*$ contains only bipartitions compatible with $\genetrees' \cup (\genetrees \setminus \genetrees') = \genetrees$, and that ${\cal B}$ contains every such bipartition.
So both algorithms consider the same set of bipartitions at each step, and lead to the same solution.
\end{proof}

It thus remain to describe how to find a core $\genetrees'$, as small as possible, as the size of the core is now the main complexity parameter. This problem is equivalent to the {\sc Minimum Set Cover Problem} known to be NP-hard. However, a natural heuristic is the following: choose a gene tree $G_i$ with the largest subset of $\genes$ as leafset, say of size $n-p$, and ``complete'' it with at most $p$ additional gene trees from $\genetrees$ each containg at least one of the missing genes. This obviously provide a core, leading to the following result.

\begin{corollary}\label{cor:complexity-minsgt-core}
The $MinSGT$ problem can be solved in $O((n+1)^{p+1} \times 4^{p+1} \times k)$ time complexity, where $p$ is the smallest integer such that a gene tree of $\genetrees$ contains $n-p$ genes.
\end{corollary}

The same technique applies to $MinLSGT$ and the same result could be stated for this problem.

\section{Triplet Respecting Supergenetrees}\label{triplets}

We now consider a problem related to the correction of a gene
tree. Assume that input gene trees $G_1, G_2, \cdots, G_k$ are
separated subtrees (i.e. leaf-disjoint) of a given gene tree
$G^{Init}$. The $MinSGT$ and $MinLSGT$ problems can also be considered
in this context to infer an alternative gene tree displaying them all
and minimizing a reconciliation cost. However, this may lead to a
new tree exhibiting a complete reorganization of the input subtrees and possibly grouping genes that were far apart in the initial tree. Therefore, assume in addition that we trust the hierarchy of
upper branches. Then we ask for a supergenetree of minimum
reconciliation cost which preserves the phylogenetic relation between
subtrees, as given by $G^{Init}$. Formally, we seek for a triplet
respecting supergenetree, as defined bellow.

\begin{definition}
  Let $\genetrees = \{G_1, G_2, \cdots, G_k \}$ be a set of separated
  subtrees of a gene tree $G^{Init}$ for $\genes$ such that
  $\bigcup_{i=1}^{k}{\L(G_i)} = \genes$. A tree $G^{TR}$ compatible
  with $\genetrees$ is {\em triplet respecting}\/ iff, for any triplet
  of trees $\L(G_{i_1})$, $\L(G_{i_2})$ and $\L(G_{i_3})$ in
  $\genetrees$ and any triplet of genes $x \in G_{i_1}$, $y \in
  G_{i_2}$ and $z \in G_{i_3}$, $G^{Init}$ and $G^{TR}$ display the
  same topology for the triplet $(x,y,z)$, i.e. $G^{Init}|_{\{x,y,z\}}
  = G^{TR}|_{\{x,y,z\}}$.
\end{definition}

\begin{figure*}[!t]
\begin{center}
\includegraphics[width= .9\textwidth]{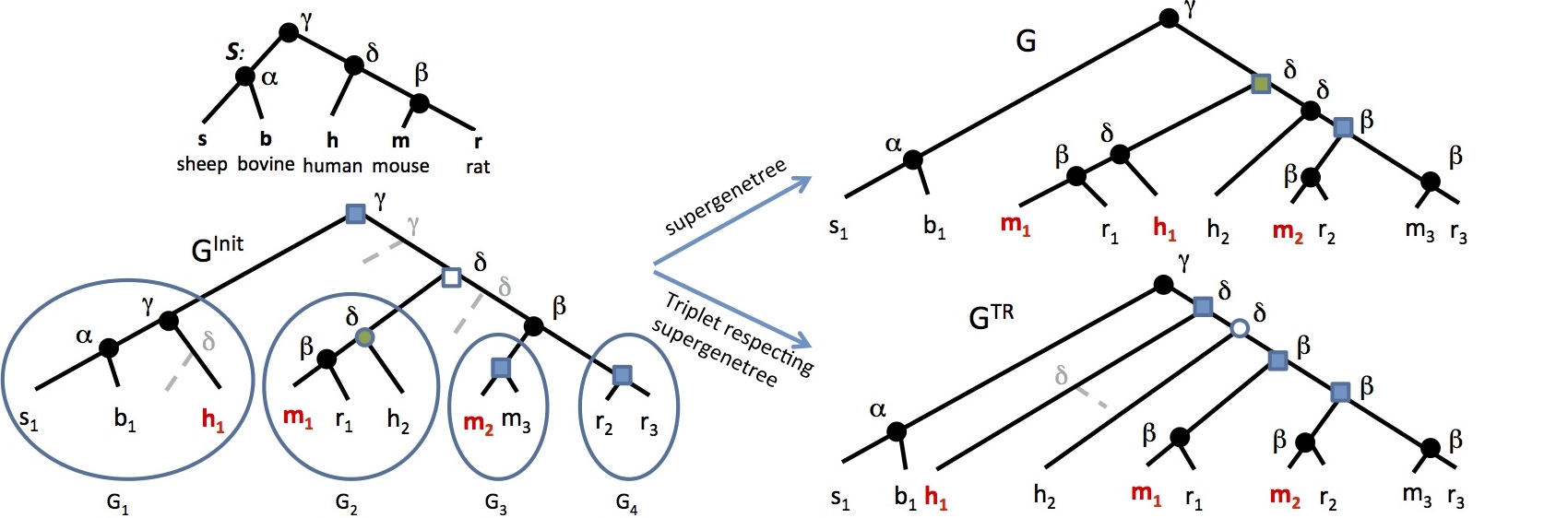}\hspace{0.1cm}
\caption{A species tree $S$ on the set of species $\Sigma =
  \{s,b,h,m,r\}$, and a gene tree $G^{Init}$ with a set of
  separated subtrees $\genetrees = \{G_1, G_2, \cdots, G_4\}$.
  The name of genes, the form and color of internal nodes and
  lines follow the same rules as in Figure 1.
  $G$ is a supergenetree for $\genetrees$ of
  minimum LCA-reconciliation cost (cost of $2$) and  $G^{TR}$ is a
  triplet respecting supergenetree for $\genetrees$ of LCA-reconciliation
  cost $4$ ($3$ duplications + $1$ loss). $G$
  is not a triplet respecting supergenetree because for example, for the
  triplet of genes $(h_1, m_1, m_2)$ in red color, $G^{Init}$ displays the
  topology $(h_1,m_1m_2)$ while $G$ displays the topology $(h_1m_1, m_2)$.
  In $G^{TR}$ however, all triplet genes topologies are respected.}\label{fig:trsupergenetree}
\end{center}
\end{figure*}

For example in Figure~\ref{fig:trsupergenetree}, the supergenetree $G$
is not triplet respecting as for the triplet of genes $(h_1, m_1,
m_2)$, $G$ does not display the same topology as the tree
$G^{Init}$.\\

\noindent \textsc{Minimum Triplet Respecting Supergenetree ($MinTRS$) Problem:}\\
\noindent {\bf Input:} A species set $\taxa$ and a species tree $S$
for $\taxa$; a gene family $\genes$ and a gene tree $G^{Init}$ for $\genes$; a set $\genetrees = \{G_1, G_2, \cdots, G_k\}$ of separated subtrees
of $G^{Init}$ such that $\bigcup_{i=1}^{k}{\L(G_i)} = \genes$.\\
\noindent {\bf Output:} Among all triplet respecting gene trees
for $\genes$ compatible with $\genetrees$, one of minimum
LCA-reconciliation cost.\\

A natural extension of the $MinTRS$ Problem is the {\sc Minimum
  Labeled Triplet Respecting Supergenetree ($MinLTRS$)} Problem, where
we are given a set of labeled separated subtrees of $G^{Init}$ and we
seek for a labeled triplet respecting supergenetree of minimum
reconciliation cost. Here we focus on $MinTRS$, though all results
extend naturally to $MinLTRS$, as briefly explained at the end of
this section.
%
%
Note that the $MinTRS$ and $MinLTRS$ problems can be reduced to the $MinSGT$ and $MinLSGT$ problems by considering as input of $MinSGT$ 
and $MinLSGT$ the set of subtrees $\genetrees$ of 
$G^{Init}$ augmented with the set of all rooted triplet trees
that should be respected by the output supergenetree. 
However, the algorithms for  $MinTRS$ and $MinLTRS$ problems induced by these reductions would remain exponential in the number of input subtrees.

We describe a more efficient recursive algorithm that solves the $MinTRS$
problem by making use of the $MinSGT$ solution. This algorithm leads to the following result.

\begin{theorem}\label{theorem:complexity-mintrs}
The $MinTRS$ and $MinLTRS$ problems can be solved in $O(n^2)$ time complexity.
\end{theorem}

The high-level description of the algorithm is as follows.
The triplet-respecting property only allows a limited 
number of ways to combine the subtrees of $\mathcal{G}$
together. We distinguish two possible cases.
First, if two subtrees $G_1, G_2$ of $\mathcal{G}$ form a ``cherry'' in $G^{Init}$, meaning that $r(G_1)$ and $r(G_2)$ share the same parent in $G^{Init}$, then $G_1$ and $G_2$ 
can be mixed together in any way without contradicting the triplet-respecting property.
The optimal way of mixing the two trees is to compute 
$MinSGT(G_1, G_2)$, which gives a solution for 
the subtree of $G^{Init}$ rooted at the parent of $r(G_1)$ and $r(G_2)$.  For instance in Figure~\ref{fig:trsupergenetree}, the two children of the $\beta$ node in $G^{Init}$ form a cherry of subtrees.  
Second, if instead a subtree $G_1$ of $\mathcal{G}$ is not part of such a cherry, then let $x$ be the sibling of $r(G_1)$ 
in $G^{Init}$.  Then we show that the following procedure can be performed: 
recursively compute $G_x$, an optimal solution for the subtree of $G^{Init}$ rooted at $x$, then try grafting 
$G_x$ on $G_1$ in every possible way and keep the solution that
minimizes the reconciliation cost.  This gives a solution 
for the subtree of $G^{Init}$ rooted at the parent of $r(G_1)$ and $x$.
These two cases describe all the possible subtree mixings that can occur, and the rest of the $G^{Init}$ topology must be conserved.
For example in Figure~\ref{fig:trsupergenetree}, 
from a bottom-up point of view,
the algorithm would compute $G_{3, 4} = MinSGT(G_3, G_4)$, 
then obtain $G_{2,3,4}$ by finding the best place on which to graft $G_{3,4}$ on $G_2$ (in this case, above the parent of $m_1$ and $r_1$), then obtain a solution by
grafting $G_{2,3,4}$ somewhere on $G_1$ (in this case above $h_1$).  In the following we rather 
describe the algorithm in a top-down manner, i.e. 
we start at the root of $G^{Init}$, obtain a solution 
recursively for its two child subtrees and combine them appropriately.

Before
describing the algorithm in full detail, we give a few additional definitions and
properties.
Let $G$ and $G'$ be two gene trees for $\genes$.  Define $s_{G'
  \rightarrow G}$ as the mapping from the nodes of $G'$ to the nodes
of $G$ such that $s_{G' \rightarrow G} (x) = lca_G(\L(x))$.  For
example in Figure~\ref{fig:trsupergenetree}, the image of the
green-colored node of $G^{Init}$ by $s_{G^{Init} \rightarrow G}$ is
the green-colored node of $G$.

The algorithm for $MinTRS$ constructs the triplet respecting
supergenetree $G^{TR}$ by building recursively and independently the
bipartitions $(\L(y_l), \L(y_r))$ induced by each internal node $y$ of
$G^{TR}$ from the root to the leaves.  The nodes of $G^{TR}$ can be
considered independently in the algorithm because the constraint of
being triplet respecting strongly predetermines the set of leaves
$\L(y)$ associated to some nodes $y$ of $G^{TR}$ as shown in Lemma
\ref{lem:lca-mapping-mintrs}.

Given a node $x$ of $G^{Init}$, we denote by $\genetrees(x)$ the
subset of $\genetrees$ that are subtrees of
$G^{Init}[x]$. If there exists a node $y$ in $G^{TR}$ such that $\L(y)
= \L(x)$, then we also define $\genetrees(y) = \genetrees(x)$.  For
example, call $x$ the white-colored node of $G^{Init}$ in Figure~\ref{fig:trsupergenetree}. Then $\genetrees(x)=\{G_2,G_3,G_4\}$. Now,
for $y$ being the white-colored node of $G^{TR}$, $\L(y)=\L(x)$ and
so $\genetrees(y)=\{G_2,G_3,G_4\}$.

\begin{lemma}
  \label{lem:lca-mapping-mintrs}
Let $G^{TR}$ be a triplet respecting supergenetree for $G^{Init}$ and
$\genetrees =
\{G_1,\ldots,G_k\}$. For any node $x$ of $G^{Init}$ such that
$|~\genetrees(x)~| \geq 2$, there exists a node $y$ of $G^{TR}$ such
that $\L(y) = \L(x)$.\\
\end{lemma}

\begin{proof}
  Let $x$ be a node of $G^{Init}$ such that  $|~\genetrees(x)~|  \geq 2$.
  Each of the subtrees $G^{Init}[x_l]$ and $G^{Init}[x_r]$
  then contains at least one tree of $\genetrees$.
  Assume that (*) there exists no node $y$ in $G^{TR}$ such that
  $\L(y) = \L(x)$.
  Let $x'$ be the node of $G^{TR}$ such that $s_{G^{Init} \rightarrow G^{TR}} (x) = x'$, and
  let $x''$ be the node of $G^{Init}$ such that $s_{G^{TR} \rightarrow G^{Init}} (x') = x''$.
  The assumption (*) implies that $x'' \neq x$, so $x''$ is a strict ancestor
  of $x$.
  Suppose w.l.o.g. that $x$ belongs to the subtree $G^{Init}[x''_l]$
  and pick any gene $c\in \L(x') \cap \L(x''_r)$. There exists
  a tree $G_h$ of $\genetrees$ such that $c\in\L(G_h)$ and $G_h$
  is contained in $G^{Init}[x''_r]$.
  Now, let $a$ and $b$ be two genes such that $a \in \L(x_l) \cap \L(x'_l)$
  and $b \in \L(x_r) \cap \L(x'_r)$, or $a \in \L(x_l) \cap \L(x'_r)$ and
  $b \in \L(x_r) \cap \L(x'_l)$. Such two genes necessarily exist because
  $s_{G^{Init} \rightarrow G^{TR}} (x) = x'$.  So, there exist two trees $G_i$ and $G_j$
  of $\genetrees$ such that $(a,b)\in\L(G_i)\times\L(G_j)$, $G_i$
  is contained in $G^{Init}[x_l]$ and $G_j$ is contained in $G^{Init}[x_r]$.
  So $G^{Init}$ displays the topology $(ab,c)$ for
  the triplet of genes $(a,b,c)$ while $G^{TR}$ displays a different topology,
  either $(a,bc)$ or $(ac,b)$. The assumption (*) is then impossible.
\end{proof}

For example, in Figure~\ref{fig:trsupergenetree}, call $x$
the white-colored node of $G^{Init}$. Then, $|~\genetrees(x)~| = 3$, and
there does exist a node $y$ in $G^{TR}$ (the white-colored node) such that
$\L(y) = \L(x)$. However, there exists no such node in $G$.

We denote by $V_{cons}(G^{TR})$ the subset of nodes $y$ of $G^{TR}$
such that there exists a node $x$ of $G^{Init}$ satisfying
$\L(x)=\L(y)$ and $|~\genetrees(y)~| \geq 2$.  For example, in Figure~\ref{fig:trsupergenetree}, $V_{cons}(G^{TR})$ contains three nodes,
the root, the white-colored node and the lowest duplication node
of $G^{TR}$. Lemma
\ref{lem:lca-mapping-mintrs} allows to predetermine the sets of leaves
$\L(y)$ associated to the nodes $y\in V_{cons}(G^{TR})$.
We now describe how to find the best subtree $G^{TR}[y]$
for each node $y\in V_{cons}(G^{TR})$,
i.e. one leading to the minimum reconciliation cost.

Note that if $y\in V_{cons}(G^{TR})$ and $|~\genetrees(y)~| = 2$,
say $\genetrees(y) = \{G_i,G_j\}$, $1\leq i < j \leq k$, then the
$MinSGT$ algorithm can be applied to build the subtree
$G^{TR}[y]$ as a minimum reconciliation cost supergenetree for $G_i$ and $G_j$.
It then remains to describe a recursive procedure for finding the subtree
$G^{TR}[y]$ for a node $y\in V_{cons}(G^{TR})$ such that $|~\genetrees(y)~| > 2$.

In order to compute the reconciliation cost of the tree $G^{TR}$,
we need to account for the local reconciliation costs for the nodes
$y\in V_{cons}(G^{TR})$, and also for the internal nodes $z$ of $G^{TR}$
such that $z\not\in V_{cons}(G^{TR})$.
To do so, given a node $y\in V_{cons}(G^{TR})$, we define
$cost_{TR}(y)$ as the local reconciliation cost for $y$, plus the
local reconciliation costs for all internal nodes $z\in V(G^{TR})$
such that $z\not\in V_{cons}(G^{TR})$, $y$ is an ancestor of $z$ and
there exists no node $y'\in V_{cons}(G^{TR})$ on the path between $y$
and $z$.
For example in Figure~\ref{fig:trsupergenetree}, call $y$
the root of $G^{TR}$. Then, $y\in V_{cons}(G^{TR})$ and $cost_{TR}(y) =
cost(\L(y_l),\L(y_r)) + cost(\L(y_{l_l}),\L(y_{l_r})) +
cost(\L(y_{r_l}),\L(y_{r_r}))$ counting the local reconciliation costs
for $y$, $y_l\not\in V_{cons}(G^{TR})$ and $y_r\not\in
V_{cons}(G^{TR})$.
%
We then obtain a formulation of the reconciliation cost of $G^{TR}$ as 
the sum of $cost_{TR}(y)$ for all nodes $y\in V_{cons}(G^{TR})$.

Lemma \ref{lem:valid-bipartition-mintrs} describes the ``valid'' configurations
of a subtree $G^{TR}[y]$ rooted at a node $y\in V_{cons}(G^{TR})$ such
that $|~\genetrees(y)~| > 2$, and the formula for computing
 $cost_{TR}(y)$  in each case. The following notations are used
 in Lemma \ref{lem:valid-bipartition-mintrs}. Given a node $x$
 of $G^{Init}$, and a node $x^*$ of $G^{Init}[x_l]$, $A(x^*)$ 
 is the set of all strict ancestors of
  $x^*$ in $G^{Init}[x_l]$,
  $A_l(x^*)$ is the subset of $A(x^*)$ such that $u\in A_l(x^*)$ if  $x^* \in V(G^{Init}[u_l])$ and $A_r(x^*)=A(x^*)\setminus A_l(x^*)$.

\begin{lemma}
  \label{lem:valid-bipartition-mintrs}
  Let $G^{TR}$ be a triplet respecting supergenetree for
  $\genetrees = \{G_1,\ldots,G_k\}$.
  Let $y$ be a node of $G^{TR}$ such that $y\in V_{cons}(G^{TR})$
  and $|~\genetrees(y)~| > 2$. Let $x$ be the node of $G^{Init}$
  such that $\L(y) = \L(x)$.

\begin{enumerate}
\item If $|~\genetrees(x_l)~| = 1$ and $|~\genetrees(x_r)~| \geq 2$,
  let $y^*\in V_{cons}(G^{TR})$ be the node of $G^{TR}$ such that
  $\L(y^*) = \L(x_r)$.
  The subtree $G^{TR}[y]$ can be obtained by taking the
  tree $G^{Init}[x_l]$ and grafting the tree $G^{TR}[y^*]$ onto it
  such that the root of $G^{TR}[y^*]$ appears as the sibling of a node
  $x^*$ of $G^{Init}[x_l]$.\\
  The cost $cost_{TR}(y)$ is then given by the following formula: 

\begin{tabular}{l}
 $cost_{TR}(y) = cost_{TR}(x_l,x_r,x^*) = $ \\
$\sum_{u~\in~A_l(x^*)}{cost(\L(u_l)\cup \L(x_r),\L(u_r))}$\\
  $+\sum_{u~\in~A_r(x^*)}{cost(\L(u_l),\L(u_r)\cup \L(x_r))}$ \\
  $+\sum_{u~\in~(V(G^{Init}[x_l]) \setminus A(x^*))}{cost(\L(u_l),\L(u_r))}$ \\
  $+cost(\L(x_l),\L(x_r))$ (if $x^* = x_l$)
\end{tabular}

%

  \item If $|~\genetrees(x_l)~| \geq 2$ and $|~\genetrees(x_r)~| = 1$,  then this case is symmetric to the previous case.

  \item If $|~\genetrees(x_l)~| \geq 2$ and $|~\genetrees(x_r)~| \geq 2$,
    then  $G^{TR}[y]$ is such that $\L(y_l)=\L(x_l)$ and $\L(y_r)= \L(x_r)$, and
$cost_{TR}(y)=cost(\L(x_l),\L(x_r))$.
\end{enumerate}

  \end{lemma}

\begin{proof}
In Case 1, we first deduce from Lemma \ref{lem:lca-mapping-mintrs}
that there must exist a node $y^*\in V_{cons}(G^{TR})$ such that
$\L(y^*) = \L(x_r)$.
Next, $G^{Init}[x_l]$ is one of the gene trees of the set $\genetrees$.
So, it must be displayed by $G^{TR}[y]$ and then, $G^{TR}[y]$ can be
obtained by taking $G^{Init}[x_l]$ and grafting $G^{TR}[y^*]$ onto it.
Finally, Case 2 is symmetric to Case 1 and 
Case 3 follows directly from Lemma \ref{lem:lca-mapping-mintrs}.
The formulas for $cost_{TR}(y)$ follows directly from the definition
of $cost_{TR}$.
\end{proof}

For example in Figure~\ref{fig:trsupergenetree}, the root and
the white-colored node of  $G^{TR}$ fulfills the conditions of the first case.
There are no node $y\in V_{cons}(G^{TR})$ satisfying $|~\genetrees(y)~| > 2$
and fulfilling the conditions of the second or third case.

We are now ready to describe the recurrence formula of the recursive algorithm
solving the $MinTRS$ problem. Given a node $x$ of $G^{Init}$ such that
$|~\genetrees(x)~| \geq 2$,
we denote by $MinTRS(G^{Init}[x])$ the minimum LCA-reconciliation
cost of a triplet respecting supergenetree compatible with $\genetrees(x)$. 

\begin{lemma}\label{theorem:mintrs}
    Let $\genetrees = \{G_1,\ldots,G_k\}$ be a set of separated subtrees
  of a gene tree $G^{Init}$ for  $\genes$ such that
  $\bigcup_{i=1}^{k}{\L(G_i)} = \genes$. Let $x$ be a node of
  $G^{Init}$ such that $|~\genetrees(x)~| \geq 2$.
  
  \begin{enumerate}
  \item (Stop condition) If $|~\genetrees(x)~|= 2$ ($\genetrees(x)=\{G_i,G_j\}$), $MinTRS(G^{Init}[x]) = MinSGT(G_i,G_j)$.
  \item Otherwise (i.e $|~\genetrees(x)~| > 2$),
    
  \begin{enumerate}
  \item If $|~\genetrees(x_l)~| = 1$ and $|~\genetrees(x_r)~| \geq 2$,
\begin{tabular}{l}
 $MinTRS(G^{Init}[x])=$ \\
$\min_{x^*~\in~ V(G^{Init}[x_l])}
    \{
  cost_{TR}(x_l,x_r,x^*)
  \}$\\
  $+ MinTRS(G^{Init}[x_r])$ \\
\end{tabular}

\item If $|~\genetrees(x_l)~| \geq 2$ and $|~\genetrees(x_r)~| = 1$,
  this case is symmetric to the previous case.
  
  \item If $|~\genetrees(x_l)~| \geq 2$ and $|~\genetrees(x_r)~| \geq 2$,
\begin{tabular}{l}
 $MinTRS(G^{Init}[x])=$ \\
$cost_{TR}(\L(x_l),\L(x_r))$\\
$+ MinTRS(G^{Init}[x_l])$\\
  $+ MinTRS(G^{Init}[x_r])$
\end{tabular}

\end{enumerate}
\end{enumerate}
  \end{lemma}  

\begin{proof}
  The proof follows from Lemmas \ref{lem:lca-mapping-mintrs}
  and \ref{lem:valid-bipartition-mintrs}, and the fact that each call
  to the recursive procedure $MinTRS$ receives as input a subtree
  $G^{Init}[x]$ such that $|~\genetrees(x)~| \geq 2$, starting with
  the whole tree rooted at $r(G^{Init})$.
  Case 1 is trivial.
  For Case 2(a) (and symmetrically Case 2(b)), following
  Lemma \ref{lem:valid-bipartition-mintrs}, 
  there are $|~V(G^{Init}[x_l])~|$ possible configurations for the subtree
$G^{TR}[y]$ rooted at $y = s_{G^{Init} \rightarrow G^{TR}}(x)$, depending on
which node $x^*$ of $G^{Init}[x_l]$ is chosen to be the sibling of
$y^* = s_{G^{Init} \rightarrow G^{TR}}(x_r)$.
Since $G^{TR}$ must be of minimum reconciliation cost, the configuration for $G^{TR}[y]$ must
be one that locally minimizes the cost $cost_{TR}(x_l,x_r,x^*)$.
Finally, Case 3 follows directly from Lemma \ref{lem:valid-bipartition-mintrs}.
\end{proof}

\noindent\textit{Complexity.} we claim that Lemma~\ref{theorem:mintrs} leads to a $O(n^2)$ 
algorithm for $MinTRS$, where $n = |V(G^{Init})|$. 
Let $x$ be a node of $G^{Init}$, let $n_x$ be the number of nodes in $G^{Init}[x]$ and let 
$n_l$ and $n_r$ be the number of nodes in the left and right subtrees of $x$, respectively.
As a base case, if $x$ falls into
case 1 of Lemma~\ref{theorem:mintrs}, then running $MinSGT$ on the two child subtrees of 
$x$ takes time $O(\max\{n_l, n_r\}^2) = O(n_x^2)$.  
Suppose instead that $x$ falls into case 2.a, and thus $\genetrees(x_l) = 1$ and $\genetrees(x_r) \geq 2$.
We may assume by induction that computing $MinTRS(G^{Init}[x_r])$ requires time $O(n_r^2)$.  
Afterwards, grafting 
the resulting tree is done on each $O(n_l)$ branch of $G^{Init}[x_l]$, and computing the cost
can be done in time $O(n_l)$ for each grafting.  Thus in total, case 2.a can be handled in time 
$O(n_r^2 + n_l^2) = O(n_x^2)$.  The case 2.b of Lemma~\ref{theorem:mintrs} is symmetric, 
and the case 2.c can be handled in constant time.  As the quadratic bound holds for every node, 
we get a bound of $O(n_x^2) = O(n^2)$ when $x$ is the root.\\


\noindent\textit {Algorithm for $MinLTRS$.}
The algorithm for the $MinTRS$ problem can be adapted to solve
the $MinLTRS$ problem. The adaptation consists in (1) replacing 
the calls to $MinSGT(G_i,G_j)$ in the stop condition of Lemma
\ref{theorem:mintrs} by calls to $MinLSGT((G_i,ev_i),(G_j,ev_j))$,
and (2) replacing the use of $cost(L_l,L_r)$ in order to define
$cost_{TR}(y)$ in Lemma \ref{lem:valid-bipartition-mintrs} by the use of
$cost_{\genetrees}(L_l,L_r)$.
Moreover,  Lemma \ref{lem:reconciliation-cost2}  must be extended such that
$cost_{\genetrees}(L_l,L_r) = +\infty$ if $(L_l,L_r)$ is not label-compatible
with $\genetrees$. The complexity of the algorithm remains unchanged in $O(n^2)$.

\section{Experiments}\label{application}

\begin{table*}[!t]
\begin{center}
\begin{tabular}{|l|c|c|c|c|} \hline
             & $\%$ of  & Avg. & Avg & Trees with\\
             & modified & running & rec. cost & better  \\ 
             & trees & time  & reduction & AU value \\ \hline
$MinSGT$ &    211  & 205240 ms  & 22.5      &  \cellcolor{gray!25} \\ 
             &    97.2\%         &         & (24.8\%)  & \cellcolor{gray!25} \\ \hline
$MinLSGT$&    207  & 113 ms & 19.5        &  \cellcolor{gray!25} \\ 
             &    95.3\%       &          & (21.5\%)  &  \cellcolor{gray!25}     \\ \hline
$MinTRS$ &  211    & 3031 ms & 15.5       &  68.6\%\\ 
             &      97.2\%      &         & (17.1\%)  &  \\ \hline
$MinLTRS$&    207  & 60 ms & 13.5       &  66.4\% \\ 
             &    95.3\%       &          & (14.9\%)  &       \\ \hline 
$PolytomySolver$&    20 ms  & 3 ms & 3.65    &  50.0\% \\ 
             &     9.2\%       &          & (4.0\%)  &        \\ \hline          
             
\end{tabular}\\
\hspace{1cm}
\caption{Results for the 217 Ensembl trees (see text for all details). Second column:
  number and percentage of corrected trees;
  Third column: mean running time for a tree in ms; Fourth column:
  mean value and percentage of the reconciliation cost reduction, i.e. difference in
  reconciliation cost between the original and corrected tree; Last column: percentage of corrected trees that have a
  better AU value than the original Ensembl trees (Due to PhyML's long computation time, 
  we could not obtain the AU values for $MinSGT$ and $MinLSGT$).}\end{center}
\end{table*}\label{table:results}
 
In the context of gene tree correction, we wanted to evaluate: (1) the benefit of the new supertree approach allowing to merge clades from different subtrees, compared with the more constrained polytomy resolution approach~\cite{LafondNoutahi2016} which conserves input subtrees separated; (2)  the benefit of the additional triplet preservation requirement of $MinTRS$ and $MinLTRS$. Both evaluations were made based on the conjecture of dubious highest duplication nodes in gene trees~\cite{HAHN07,SwensonMabrouk12}.

For this purpose, we considered the gene trees of the Ensembl vertebrate database Release 84 rooted at a duplication node. For each tree $G^{Init}$, $\genetrees$ was defined as the set of all subtrees of $G^{Init}$ rooted at the ``highest speciation nodes'', i.e. speciation nodes with only duplication nodes as ancestors. On average, $G^{Init}$ contains $121.2$ leaves and is partitioned into  $7.5$ subtrees.  Aiming at comparing all developed algorithms, including the exponential time $MinSGT$ and $MinLSGT$, we restricted the sample to the $217$ gene trees with at most $200$ leaves and partitioned into at most $5$ subtrees. We also applied, on these $217$ trees,  $PolytomySolver$~\cite{LafondNoutahi2016}  which, given a set of trees $G_1, \ldots, G_k$, finds a binary tree with leafset $\{G_1, \ldots, G_k\}$ such that the reconciliation cost of the resulting tree is minimum.  Results are given in Table~\ref{table:results}. 

While the four supertree algorithms correct more than $200$ trees, corresponding to more than $95\%$ of the 217 trees, $PolytomySolver$ only corrects $20$ trees corresponding to about $9\%$ of the trees. 
Additionally, $PolytomySolver$ reduces the reconciliation cost by only $4\%$ on average on the $20$ corrected trees, compared to more than $15\%$ for supertree algorithms. Clearly, by exploring a larger solution space, $MinSGT$ and $MinLSGT$ allow to obtain the best solutions in terms of reconciliation cost. 


As for $MinTRS$ and $MinLTRS$, although more constrained than $MinSGT$ and $MinLSGT$, they lead surprisingly to almost as much correction as these two algorithms, while the correction achieved by $PolytomySolver$ is clearly less. The triplet preserving constraint appears to be less stringent than the conservation of the subtrees. In particular, for the  trees leading to only two subtrees, $PolytomySolver$ conserves the initial tree. 
Notice that introducing the labeling constraint ($MinSGT$ versus $MinLSGT$ and $MinTRS$ versus $MinLTRS$) only leads to a slight decrease of the  correction rates. 

Finally, in order to assess the benefit of the triplet respecting constraint and the quality of the correction achieved, the trees corrected by  $MinTRS, MinLTRS$ and $PolytomySolver$ were evaluated according to their
statistical support. PhyML~\cite{guindon2003simple} was executed to obtain the
log-likelihood values per site
(note that $150$ trees were included in this evaluation, as PhyML was very time-consuming on the larger trees).  
Consel~\cite{CONSEL} was then run to evaluate, using the AU (Approximately Unbiased) test, 
if the likelihood differences of pairs of  Ensembl and corrected gene trees were significant enough to statistically reject one of them.
A tree can be
rejected if its AU value, interpreted as a p-value, is under
$0.05$. Otherwise, no significant evidence allows us to reject one of
the two trees. 

Interestingly, when compared to the tree output by 
$MinTRS$ (respectively $MinLTRS$), $48.5$\% (resp. $46.5$\%) of the Ensembl gene trees are
rejected compared to only $11.9$\% (resp. $11.0$\%) of the corrected gene trees.
More than $68.6$\% (resp. $66.4$\%) of the corrected trees have better AU values than original trees. As for PolytomySolver, $5$\% of the Ensembl trees were rejected, as $25$\% of the corrected trees were rejected, with $50$\% of the corrected trees obtaining a better AU value.  The performance of $MinTRS$ and $MinLTRS$ is rather surprising as our
correction, based on the phylogenetic information of the species tree,
is not expected to improve tree likelihood based on sequence
similarity. This may be an indication that high duplications are
actually dubious and that a correction specifically focusing
on such duplications is able to significantly improve the accuracy of
the tree. This observation is further supported  by the fact that the number of highest
duplications is lower for corrected trees than for initial trees (data not shown), showing that our
correction algorithms have the general tendency of deleting high duplications.


\section{Conclusion}

This paper introduces a new methodology combining the supertree and
reconciliation frameworks with the purpose of constructing a gene tree by
combining a set of trees on partial, possibly overlapping data. We also
show how this new paradigm is  useful for gene tree correction. In particular, the
artifact of duplications wrongly inferred close to the root of a gene
tree has been reported in the literature. Here, we propose a new method for
correcting a gene tree, by first removing the higher duplication nodes and then finding the supertree best fitting the species
tree, that preserves the remaining ``trusted'' subtrees, and possibly their
hierarchical position in the initial gene tree. This supertree approach is shown to correct more trees
than the approach based on resolving a polytomy, as the
first correction allows the clustering of genes from different
input subtrees. The  corrected Ensembl gene trees are shown to
exhibit less highest duplication nodes and a lower reconciliation cost. Corrected gene trees are also shown to have a better likelihood support.

This new gene tree construction and correction paradigm leads to many
new open problems. In particular, no proof currently exists on the
complexity of the problem of finding a supergenetree minimizing the
reconciliation cost, although it is likely to be NP-hard, based on the fact
that minimizing the duplication cost is hard.
The two problems (reconciliation versus duplication costs) probably
also share the same inapproximability properties.  However, it is
possible that the supertree problems presented here are
fixed-parameter tractable with respect to parameters such as
the number of trees, the minimum reconciliation cost or the size of
the intersection between the leafset of the trees.  This is an area
that deserves a more in-depth investigation.
In addition, while
the extension to the labeled case has been done with the same
exponential complexity, adding the label restriction strongly
constrains the set of explored bipartitions, and we can expect a more
efficient algorithm in this case.

The problems we consider are build upon strong underlying assumptions, such as the
consistency of input gene trees, the compatibility and
$S$-consistency of input gene relations. A natural extension is then
to integrate the notion of a minimal correction of input trees to fit
these preliminary conditions. Finally, from an application point of
view, rather than removing higher duplication nodes, other types of gene tree pruning 
can be envisaged and used to select the initial ``trusted''
phylogenetic information that can then be combined using our
supergenetree and reconciliation framework.

\bibliographystyle{IEEEtran}
\bibliography{superGeneTree}      

\begin{thebibliography}{10}
\providecommand{\url}[1]{#1}
\csname url@samestyle\endcsname
\providecommand{\newblock}{\relax}
\providecommand{\bibinfo}[2]{#2}
\providecommand{\BIBentrySTDinterwordspacing}{\spaceskip=0pt\relax}
\providecommand{\BIBentryALTinterwordstretchfactor}{4}
\providecommand{\BIBentryALTinterwordspacing}{\spaceskip=\fontdimen2\font plus
\BIBentryALTinterwordstretchfactor\fontdimen3\font minus
  \fontdimen4\font\relax}
\providecommand{\BIBforeignlanguage}[2]{{%
\expandafter\ifx\csname l@#1\endcsname\relax
\typeout{** WARNING: IEEEtran.bst: No hyphenation pattern has been}%
\typeout{** loaded for the language `#1'. Using the pattern for}%
\typeout{** the default language instead.}%
\else
\language=\csname l@#1\endcsname
\fi
#2}}
\providecommand{\BIBdecl}{\relax}
\BIBdecl

\bibitem{Bininda04}
O.~Bininda-Emonds, Ed., \emph{Phylogenetic Supertrees combining information to
  reveal The Tree Of Life}, ser. Computational Biology.\hskip 1em plus 0.5em
  minus 0.4em\relax Dordrecht, the Netherlands: Kluwer Academic, 2004.

\bibitem{Bansal10}
M.~Bansal, J.~Burleigh, O.~Eulenstein, and D.~Fern\'andez-Baca,
  ``Robinson-foulds supertrees,'' \emph{Alg. Mol. Biol.}, vol.~5, no.~18, 2010.

\bibitem{Warnow12}
N.~Nguyen, S.~Mirarab, and T.~Warnow, ``{MRL} and {SuperFine+MRL}: new
  supertree methods,'' \emph{Alg. Mol. Biol.}, vol.~7, no.~3, 2012.

\bibitem{PhySIC07}
V.~Ranwez, V.~Berry, A.~Criscuolo, P.~Fabre, S.~Guillemot, C.~Scornavacca, and
  E.~Douzery, ``{PhySIC}: a veto supertree method with desirable properties,''
  \emph{Syst. Biol.}, vol.~56, no.~5, pp. 798\-- 817, 2007.

\bibitem{Douzery10}
V.~Ranwez, A.~Criscuolo, and E.~Douzery, ``{SuperTriplets}: a triplet-based
  supertree approach to phylogenomics,'' \emph{Bioinformatics}, vol.~26,
  no.~12, pp. i115\-- i123, 2010.

\bibitem{Steel08}
M.~Steel and A.~Rodrigo, ``Maximum likelihood supertrees,'' \emph{Syst. Biol.},
  vol.~57, no.~2, pp. 243\-- 250, 2008.

\bibitem{Warnow12b}
M.~Swenson, R.~Suri, C.~Linder, and T.~Warnow, ``{SuperFine}: fast and accurate
  supertree estimation,'' \emph{Sys. Biol.}, vol.~61, no.~2, pp. 214\--227,
  2012, {S}pecial issue RECOMB-CG 2012.

\bibitem{Steel92}
M.~Steel, ``The complexity of reconstructing trees from qualitative characters
  and subtrees,'' \emph{J. Classif.}, vol.~9, pp. 91 \-- 116, 1992.

\bibitem{Scornavacca14}
C.~Scornavacca, L.~van Iersel, S.~Kelk, and D.~Bryant, ``The agreement problem
  for unrooted phylogenetic trees is {FPT},'' \emph{J. Graph Algorithms Appl.},
  vol.~18, no.~3, pp. 385 \-- 392, 2014.

\bibitem{Aho81}
A.~Aho, S.~Yehoshua, T.~Szymanski, and J.~Ullman, ``Inferring a tree from
  lowest common ancestors with an application to the optimization of relational
  expressions,'' \emph{SIAM J. Comput.}, vol.~10, no.~3, pp. 405\-- 421, 1981.

\bibitem{Sankoff95}
M.~Constantinescu and D.~Sankoff, ``An efficient algorithm for supertrees,''
  \emph{J. Classif.}, vol.~12, pp. 101\-- 112, 1995.

\bibitem{Ng96}
M.~Ng and N.~Wormald, ``Reconstruction of rooted trees from subtrees,''
  \emph{Discrete Appl. Math}, vol.~69, pp. 19\-- 31, 1996.

\bibitem{Semple03}
C.~Semple, ``Reconstructing minimal rooted trees,'' \emph{Discrete Appl.
  Math.}, vol. 127, no.~3, 2003.

\bibitem{OrthoMCL03}
L.~Li, C.~J. Stoeckert, and D.~Roos, ``{OrthoMCL}: identification of ortholog
  groups for eukaryotic genomes,'' \emph{Genome Res.}, vol.~13, pp. 2178\--
  2189, 2003.

\bibitem{InParanoid08}
A.~Berglund, E.~Sjolund, G.~Ostlund, and E.~Sonnhammer, ``{InParanoid} 6:
  eukaryotic ortholog clusters with inparalogs,'' \emph{Nucleic Acids Res.},
  vol.~36, pp. D263 \-- D266, 2008.

\bibitem{Proteinortho11}
M.~Lechner, S.~Findeib, L.~Steiner, M.~Marz1, P.~Stadler, and S.~Prohaska,
  ``Proteinortho: detection of (co-)orthologs in large-scale analysis,''
  \emph{BMC Bioinformatics}, vol.~12, p. 124, 2011.

\bibitem{HAHN07}
M.~Hahn, ``Bias in phylogenetic tree reconciliation methods: implications for
  vertebrate genome evolution,'' \emph{Genome Biol.}, vol.~8, no. R141, 2007.

\bibitem{TREEFIX}
Y.-C. Wu, M.~D. Rasmussen, M.~S. Bansal, and M.~Kellis, ``{TreeFix:
  Statistically Informed Gene Tree Error Correction Using Species Trees},''
  \emph{Syst. Biol.}, vol.~62, no.~1, pp. 110--120, 2013.

\bibitem{CHEN-JCOMPBIOL7}
K.~Chen, D.~Durand, and M.~Farach-Colton, ``Notung: Dating gene duplications
  using gene family trees,'' \emph{J. Computat. Biol.}, vol.~7, pp. 429--447,
  2000.

\bibitem{Zheng-Zhang2014}
Y.~Zheng and L.~Zhang, ``Reconciliation with non-binary gene trees revisited,''
  in \emph{Proceedings of RECOMB 2014}, ser. Lecture Notes Comput. Sci., vol.
  8394, 2014, pp. 418\--432, proceedings of RECOMB.

\bibitem{SwensonMabrouk12}
K.~M. Swenson and N.~El-Mabrouk, ``Gene trees and species trees: Irreconcilable
  differences,'' \emph{BMC Bioinformatics}, vol.~13, no. (Suppl 19), p. S15,
  2012.

\bibitem{TREEFIXDTL}
M.~S. Bansal, Y.~Wu, E.~J. Alm, and M.~Kellis, ``Improved gene tree error
  correction in the presence of horizontal gene transfer,''
  \emph{Bioinformatics}, vol.~31, no.~8, pp. 1211--1218, 2015.

\bibitem{noutahi2016}
E.~Noutahi, M.~Semeria, M.~Lafond, J.~Seguin, B.~Boussau, L.~Guéguen,
  N.~El-Mabrouk, and E.~Tannier, ``Efficient gene tree correction guided by
  genome evolution,'' \emph{Plos One}, 2016, to appear.

\bibitem{LafondNoutahi2016}
M.~Lafond, E.~Noutahi, and N.~El-Mabrouk, ``{Efficient Non-Binary Gene Tree
  Resolution with Weighted Reconciliation Cost},'' in \emph{27th Annual
  Symposium on Combinatorial Pattern Matching (CPM 2016)}, ser. Leibniz
  International Proceedings in Informatics (LIPIcs), vol.~54, 2016, pp.
  14:1--14:12.

\bibitem{Liberles08}
S.~Massey, A.~Churbanov, S.~Rastogi, and D.~Liberles, ``Characterizing positive
  and negative selection and their phylogenetic effects,'' \emph{Gene}, vol.
  418, pp. 22\-- 26, 2008.

\bibitem{Skovgaard06}
M.~Skovgaard, J.~Kodra, D.~Gram, S.~Knudsen, D.~Madsen, and D.~Liberles,
  ``Using evolutionary information and ancestral sequences to understand the
  sequence-function relationship in {GLP}-1 agonists,'' \emph{J. Mol. Biol.},
  vol. 363, pp. 977\-- 988, 2006.

\bibitem{Taylor05}
S.~Taylor, K.~de~la Cruz, M.~Porter, and M.~Whiting, ``Characterization of the
  long-wavelength opsin from {M}ecoptera and {S}iphonaptera: does a flea see?''
  \emph{Mol. Biol. Evol.}, vol.~22, pp. 1165\-- 1174, 2005.

\bibitem{Lafond15}
M.~Lafond, A.~Ouangraoua, and N.~El-Mabrouk, ``Reconstructing a supergenetree
  minimizing reconciliation,'' \emph{BMC Genomics}, vol.~16, p.~S4, 2015,
  {S}pecial issue of RECOMB-CG 2015.

\bibitem{Fitch2000}
W.~M. Fitch, ``Homology. a personal view on some of the problems,''
  \emph{Trends Genet.}, vol.~16, no.~5, pp. 227\-- 231, 2000.

\bibitem{lafond2014orthology}
M.~Lafond and N.~El-Mabrouk, ``Orthology and paralogy constraints:
  satisfiability and consistency,'' \emph{BMC Genomics}, vol.~15, no. Suppl 6,
  p. S12, 2014, {S}pecial issue RECOMB-CG 2014.

\bibitem{Zhang97}
L.~Zhang, ``On a {M}irkin-{M}uchnik-{S}mith conjecture for comparing molecular
  phylogenies,'' \emph{J. Comput. Biol.}, vol.~4, no.~2, pp. 177\-- 187, 1997.

\bibitem{CHAUVE-ELMABROUK09}
C.~Chauve and N.~El-Mabrouk, ``New perspectives on gene family evolution:
  losses in reconciliation and a link with supertrees,'' in \emph{Proceedings
  of RECOMB 2009}, ser. Lecture Notes Comput. Sci., vol. 5541.\hskip 1em plus
  0.5em minus 0.4em\relax Springer, 2009, pp. 46\--58.

\bibitem{Hellmuth-2013}
M.~Hellmuth, M.~Hernandez-Rosales, K.~Huber, V.~Moulton, P.~Stadler, and
  N.~Wieseke, ``Orthology relations, symbolic ultrametrics, and cographs,''
  \emph{J. Math. Biol.}, vol.~66, no. 1--2, pp. 399--420, 2013.

\bibitem{guindon2003simple}
S.~Guindon and O.~Gascuel, ``A simple, fast, and accurate algorithm to estimate
  large phylogenies by maximum likelihood,'' \emph{Syst. Biol.}, vol.~52,
  no.~5, pp. 696--704, 2003.

\bibitem{CONSEL}
H.~Shimodaira and M.~Hasegawa, ``{CONSEL}: for assessing the confidence of
  phylogenetoc tree selection,'' \emph{Bioinformatics}, vol.~17, pp. 1246\--
  1247, 2001.

\end{thebibliography}

\end{document}